%% file: caqec.tex
\documentclass[prl,letterpaper,twocolumn]{revtex4}

\usepackage{amsmath,amsbsy,amsfonts,hyperref,bbm,verbatim,amsthm,amssymb}

\newcommand{\ket}[1]{\left|#1\right\rangle}
\newcommand{\cket}[1]{\left|\widetilde{#1}\right\rangle}
\newcommand{\bra}[1]{\left\langle #1\right|}

\newcommand{\ketbra}[1]{\ket{#1}\bra{#1}}

\newcommand{\bracket}[2]{\langle #1|#2\rangle}

\newtheorem{theorem}{Theorem}
\newtheorem{corollary}{Corollary}
\usepackage{tikz}
\usetikzlibrary{backgrounds,fit,decorations.pathreplacing,positioning}

\begin{document}
\author{Joseph M.~Renes}
\affiliation{Institut f\"ur Angewandte Physik, Technische Universit\"at Darmstadt, Hochschulstr.~4a, 64289 Darmstadt, Germany}

\title{Approximate Quantum Error Correction via Complementary Observables}

\begin{abstract}
The breakthrough of quantum error correction brought with it the picture of quantum information as a sort of combination of two complementary types of classical information, ``amplitude'' and ``phase''. Here I show how this intuition can be used to construct two new conditions for approximate quantum error correction. The first states that entanglement is locally recoverable from a bipartite state when one system can be used to approximately predict the outcomes of two complementary observables on the other. The second, more in the spirit of the recent decoupling approach, states that entanglement is locally recoverable when the environment cannot reliably predict either.\\  
This paper has been superceded by \href{http://arxiv.org/abs/1605.01420}{{\color{blue} arXiv:1605.01420}}.
\end{abstract}

\maketitle

In recent years \emph{decoupling} has emerged as the standard approach to constructing quantum information processing protocols. Originally due to Schumacher and Westmoreland~\cite{schumacher_approximate_2002}, decoupling gives a simple condition on when local operations on a bipartite system $\psi^{AB}$ are sufficient to transform it into a maximally entangled state, stating that this is possible when one of the systems, $A$ say, is approximately uncorrelated with the environment $R$ and nearly in a completely random state, in the sense that $\psi^{AR}\approx \frac{1}{d}\mathbbm{1}^A\otimes\psi^R$. Here the environment refers to the system purifying $\psi^{AB}$, so that $\ket{\psi}^{ABR}$ is a pure state, while $\psi^{AR}$ denotes the marginal state of $AR$. Treating $\psi^{AB}$ as the output of a channel, this translates into a condition on recovery of entanglement from the channel.

Protocols for both entanglement distillation and quantum communication over noisy channels are relatively simple to construct using the decoupling condition, which has led to large number of results~\cite{klesse_approximate_2007,hayden_decoupling_2008,hayden_random_2008,horodecki_quantum_2008,
klesse_random_2008,bjelakovi_entanglement_2009,buscemi_one-shot_2009}. Moreover, decoupling also offers an easy approach to state merging, the process by which local operations and classical communication from $A$ to $B$ leave $B$ with the purification of $R$~\cite{horodecki_partial_2005,horodecki_quantum_2007}. Building up from state merging and entanglement distillation, decoupling has been used to construct an entire ``family tree'' of quantum protocols~\cite{abeyesinghe_mother_2009}.

But what of the original intuition that quantum information is a sort of combination of classical ``amplitude'' and ``phase'' information? This picture stems from the digitization of quantum errors into classical amplitude and phase errors ($\sigma_z$ and $\sigma_x$ errors for qubit systems, respectively), and the construction of quantum error-codes by combining classical error-correcting codes for each of these error types. Many of the earliest quantum information processing protocols were built around this intuition~\cite{bennett_mixed-state_1996,deutsch_quantum_1996,lo_unconditional_1999,shor_simple_2000}, and it also relates quantum information processing to foundational questions, as the two types of errors are associated with complementary observables. 

However, protocols which exploit the digitization of quantum errors in this manner are suboptimal. Entanglement distillation with one-way communication provides a simple example. There, a large number of copies of an arbitrary quantum state $\psi^{AB}$ can be transformed into maximally entangled pairs $\ket{\Phi_d}^{AB}=\frac{1}{d}\sum_{z=0}^{d-1}\ket{z,z}^{AB}$, where $d={\rm dim}(A)$, using local operations and classical communication from $A$ to $B$ at rate $E_\rightarrow(\psi^{AB})=H(B)_\psi-H(AB)_\psi$, for $H$ the von Neumann entropy~\cite{devetak_distillation_2005}. But a protocol based on quantum error-correcting codes effectively sees only the part of $\psi^{AB}$ diagonal in the basis of Bell states $\ket{\beta_{jk}}^{AB}=\left(X^jZ^k\right)^B\ket{\Phi_d}^{AB}$, leading to a suboptimal rate. Here $X$ and $Z$ are 
defined by $Z=\sum_k \omega^k\ketbra{k}$ and $X=\sum_k\ket{k{+}1}\bra{k}$, for $\omega=e^{2\pi i/d}$. 

In this paper I show that, despite these obstacles, the intuition of quantum information as a combination of complementary classical information is as powerful as the decoupling approach. More specifically, I establish two new conditions for entanglement recovery from $\psi^{AB}$, both of which also imply $B$ alone contains the purification of $R$. Either can then be taken as the basis for constructing the protocols of the quantum family tree.

The first condition states that Bob, who holds $B$, can perform a unitary operation on $U^{BCD}$ on $B$ and ancillary systems $C$ and $D$ and approximately recover entanglement in the form of $\ket{\Phi_d}^{AB}$ when, using measurements $\mathcal{M}_X$ or $\mathcal{M}_Z$, he can approximately predict the outcome of measuring the two complementary observables $X^A$ and $Z^A$ on $A$, held by Alice. Fig.~\ref{fig:entdec} shows the construction of $U^{BCD}$ from $\mathcal{M}_X$ and $\mathcal{M}_Z$.
\input{entdecfigure}


The second condition shows that the decoupling condition extends to classical complementary observables in that entanglement is recoverable from $\psi^{AB}$ when, roughly speaking, the purifying system $R$ cannot be used to predict either of Alice's two measurements. This is accomplished by appealing to a recent result linking $R$'s ability to predict $X^A$ ($Z^A$) with Bob's ability to predict the complementary $Z^A$ ($X^A$)~\cite{renes_duality_2010}, and then using the recovery operation from the first condition. 

The two conditions for approximate quantum error correction 
are stated and proven more precisely as Theorems 1 and 2. Afterwards, I discuss their applications and relation to other work. 

\emph{Entanglement Recovery.}---Before delving into the rigorous results, it is worth elaborating on an important aspect of this approach. As shown schematically in Fig.~\ref{fig:entdec}, the recovery operator is constructed by using Bob's measurements $\mathcal{M}_Z$ and $\mathcal{M}_X$ in sequence. The first measurement  creates a local (approximate) copy of the $Z^A$ observable in an ancillary system $C$, in the following sense. Without loss of generality, the initial state is of the form 
\begin{align}
\label{eq:geninput}
\ket{\psi}^{ABR}=\sum_{z=0}^{d-1} \sqrt{p_z}\ket{z}^A\ket{\varphi_z}^{BR},
\end{align}
for some probability distribution $p_z$ and normalized states $\ket{\varphi_z}^{BR}$. If the first measurement were perfect, it would create the state 
\begin{align}
\label{eq:zextension}
\ket{\psi_Z}^{ABCR}=\sum_{z=0}^{d-1}\sqrt{p_z}\ket{z}^A\ket{z}^C\ket{\varphi_z}^{BR},
\end{align}
so that $C$ is just a copy of $A$ in the $Z^A$ basis. This state will be referred to as the \emph{$Z^A$-extension} of $\psi^{AB}$ to $C$. 

The crucial point is then that it is not strictly necessary for the second measurement to be made on $B$ alone. Instead, $\mathcal{M}_X$ could also involve $C$. And indeed, this turns out to be critical to the optimality of the present approach, because it is the means by which Bob can exploit possible correlations between $Z^A$ and $X^A$, as well as avoid the problems associated with digitization as mentioned above.

The same point is also critical to the second condition. Simply requiring the environment $R$ itself to be uncorrelated with the $X^A$ and $Z^A$ outcomes is not enough to ensure entanglement between $A$ and $B$. For a very simple example, consider the tripartite qubit state $\ket{\psi}^{ABR}=\frac{1}{\sqrt{2}}\left(\ket{0_y,0_y,0_y}+\ket{1_y,1_y,1_y}\right)^{ABR}$, where $\ket{k_y}$ is the eigenstate of $\sigma_y=i\sigma_x\sigma_z$ with eigenvalue $(-1)^k$. Although $R$ is completely uncorrelated with measurement of both $X=\sigma_x$ and $Z=\sigma_z$ on Alice's system, $\psi^{AB}$ is not a maximally entangled state. 
As we shall see, one way to fix this is to strengthen the decoupling condition for one of the observables, say $X^A$, by requiring that even measurement of $R$ and $C$ cannot predict $X^A$, where again $C$ comes from the $Z^A$ extension of $\psi^{AB}$. 

Now we are ready to proceed to the precise formulation of the results. To do so, we first need to properly define the various ``approximate'' notions which will be used. For approximate entanglement, a good option is to use the trace distance between the actual state $\psi^{AB}$ and the ideal $\Phi_d^{AB}$: If $\left\|\psi^{AB}-\Phi_d^{AB}\right\|_1$ is small, $\psi^{AB}$ is approximately entangled. Here $\left\|M\right\|_1={\rm Tr}\sqrt{M^\dagger M}$ for any operator $M$. 

The predictability of an observable $Z^A$ given measurement of a system $B$, which is used in the first condition, can be quantified by the maximum guessing probability over all measurements $\mathcal{M}_Z$ with elements $\Lambda_z^B$:
\begin{align}
p_{\rm guess}(Z^A|B)_\psi=\max_{\mathcal{M}_Z}\sum_{z=0}^{d-1} p_z {\rm Tr}\left[\Lambda^B_z\varphi^B_z\right],
\end{align}
using the general form of the state $\psi^{AB}$ from (\ref{eq:geninput}). 

Finally, for the second condition we shall need to quantify the extent to which an observable $Z^A$ is both completely unknown to $R$ and uniformly distributed. For a general state $\psi^{AR}$ the ideal case would therefore be $\tfrac{1}{d}\mathbbm{1}^A\otimes \psi^R$, and the fidelity of the two turns out to be a convenient quantity. We define
\begin{align}
p_{\rm secure}(Z^A|B)_\psi=\sum_{z=0}^{d-1}\sqrt{\tfrac{p_z}{d}}F(\varphi_z^R,\psi^R),
\end{align}
where $F(\rho,\sigma)=\left\|\sqrt{\rho}\sqrt{\sigma}\right\|_1$. 

With these definitions in place, we now turn to the first condition.
\begin{theorem}
\label{thm:entdec}
Given a state $\ket{\psi}^{ABR}$ and $Z^A$-extension $\ket{\psi_Z}^{ABCR}$, with $d={\rm dim}(A)$, suppose $p_{\rm guess}(Z^A|B)_\psi>1-\epsilon_z$ and $p_{\rm guess}(X^A|BC)_{\psi_Z}>1-\epsilon_x$. Then there exists an isometry $U^{B\rightarrow BCD}$ such that 
\[\left\|\ket{\Phi_d}^{AD}\ket{\psi}^{CBR}-U^{B\rightarrow BCD}\ket{\psi}^{ABR}\right\|_1\leq \sqrt{2\epsilon_x}+\sqrt{2\epsilon_z}.\] 
\end{theorem}
As discarding a system can only decrease the guessing probability, we immediately have the corollary that  $p_{\rm guess}(Z^A|B)_\psi>1-\epsilon_z$ and $p_{\rm guess}(X^A|B)_\psi>1-\epsilon_x$ implies the same result.
\begin{proof}
Call the measurement maximizing the $Z^A$ guessing probability $\Lambda^B_z$, and start by implementing this measurement coherently with an isometry $U_1^{B\rightarrow BC}$, storing the result of the measurement in system $C$. This produces the state $\ket{\psi_1}^{ABCR}\equiv U_1^{B\rightarrow BC}\ket{\psi}^{ABR}$, with without loss of generality takes the form 
\begin{align}
\ket{\psi_1}^{ABCR}=\sum_{z,z'}\sqrt{p_z}\ket{z}^A\ket{z'}^C\sqrt{\Lambda^B_{z'}}\ket{\varphi_z}^{BR}.
\end{align}
Using the fact that $\sqrt{\Lambda}\geq \Lambda$ for $0\leq \Lambda\leq \mathbbm{1}$, the first condition directly implies that $\bracket{\psi_Z}{\psi_1}\geq 1-\epsilon_z$. 

Thus we may pretend that the output of the first step is $\ket{\psi_Z}^{ABCR}$. Converting the $A$ system to the $X^A$ basis, with elements $\cket{x}=\frac{1}{\sqrt{d}}\sum_z \omega^{xz}\ket{z}$, gives $\ket{\psi_Z}=\tfrac{1}{\sqrt{d}}\sum_x \cket{x}^A \left(Z^{-x}\right)^C\ket{\psi}^{CBR}$, where $\ket{\psi}^{CBR}$ is just $\ket{\psi}^{ABR}$ with system $A$ replaced with $C$. 
Now perform the measurement for guessing $X^A$ coherently with the isometry $U_2^{BC\rightarrow BCD}$, storing the result in $D$; call the output $\ket{\psi_Z'}^{ABCDR}\equiv U_2^{BC\rightarrow BCD}\ket{\psi_Z}^{ABCR}$. The ideal output would be 
\begin{align}
\ket{\xi}^{ABCDR}=\tfrac{1}{\sqrt{d}}\sum_x \cket{x}^A\ket{-\widetilde{x}}^D \left(Z^{-x}\right)^C\ket{\psi}^{CBR},
\end{align}
and applying the same reasoning as for $U_1^{B\rightarrow BC}$, it follows that $\bracket{\xi}{\psi'_Z}\geq 1-\epsilon_x$. Furthermore, a controlled-phase operation $U_3^{CD}=\sum_x \ketbra{\widetilde{x}}^{D}\otimes \left(Z^x\right)^C$ applied to $\ket{\xi}$ gives $\ket{\Phi_d}^{AD}\ket{\psi}^{CBR}$, since $\ket{\Phi_d}^{AB}=\frac{1}{\sqrt{d}}\sum_x \cket{x}^A\ket{-\widetilde{x}}^B$.

Finally, define $U^{B\rightarrow BCD}=U_3^{CD}U_2^{BC\rightarrow BCD}U_1^{B\rightarrow BC}$ and compute the trace distance between $\ket{\Phi_d}^{AD}\ket{\psi}^{CBR}$ and $U^{B\rightarrow BCD}\ket{\psi}^{ABR}$. When the fidelity of two states exceeds $1-\epsilon$, their trace distance is not larger than $\sqrt{2\epsilon}$. Therefore,  $\left\|\ket{\Phi_d}^{AD}\ket{\psi}^{CBR}-U_3U_2\ket{\psi_Z}^{ABCR}\right\|_1\leq \sqrt{2\epsilon_x}$  by unitary invariance of the trace distance. Since $\left\|\ket{\psi_Z}^{ABCR}-U_1\ket{\psi}^{ABR}\right\|_1\leq \sqrt{2\epsilon_z}$,  using the triangle inequality completes the proof. \end{proof}

Note that in the proof, the operators $U_i$ are formulated as isometries from one state space to another, whereas in Fig.~\ref{fig:entdec} they take the form of unitaries. The latter is accomplished by explicitly including the ancilla systems $C$ and $D$ from the start, which is avoided in the proof to reduce clutter. 

Using the recent result that $X^A$ being decorrelated from $RC$ is dual to $Z^A$ being correlated with $B$, the second condition follows immediately from the first.  
\begin{theorem}
\label{thm:decouple}
Suppose $p_{\rm secure}(X^A|CR)_{\psi_Z}>1-\epsilon_x$ and $p_{\rm secure}(Z^A|R)_\psi>1-\epsilon_z$ for a state $\ket{\psi}^{ABR}$ and $Z^A$-extension $\ket{\psi_Z}^{ABCR}$, with $d={\rm dim}(A)$. Then there exists an isometry $U^{B\rightarrow BCD}$ such that 
\[\left\|\ket{\Phi_d}^{AD}\ket{\psi}^{CBR}-U^{B\rightarrow BCD}\ket{\psi}^{ABR}\right\|_1\leq \sqrt[4]{8\epsilon_x}+\sqrt[4]{8\epsilon_z}.\] 
\end{theorem}
\begin{proof}
Observe from Eq.~(\ref{eq:zextension}) that $p_{\rm guess}(Z^A|BC)_{\psi_Z}=1$, or equivalently, $H(Z^A|BC)=0$, using the conditional von Neumann entropy. Theorem 2 of~\cite{renes_duality_2010} then states that $p_{\rm secure}(X^A|CR)_{\psi_Z}>1-\epsilon_x$ implies $p_{\rm guess}(Z^A|B)_{\psi_Z}>1-\sqrt{2\epsilon_x}$. Since 
$p_{\rm guess}(Z^A|B)_{\psi_Z}=p_{\rm guess}(Z^A|B)_{\psi}$, we have the first premise of Theorem 1, with approximation parameter $\sqrt{2\epsilon_x}$.  
By the same reasoning, $p_{\rm secure}(Z^A|R)_{\psi}>1-\epsilon_z$ implies $p_{\rm guess}(X^A|BC)_{\psi_Z}>1-\sqrt{2\epsilon_z}$. Applying Theorem~\ref{thm:entdec} then gives the desired result. 
\end{proof}

\emph{Discussion and Applications.---} The first condition was implicitly used in~\cite{renes_physical_2008} to construct optimal entanglement distillation protocols using one-way classical communication and extended in~\cite{boileau_optimal_2009} to the task of state merging. Results for the ``mother of all protocols'', the fully-quantum Slepian-Wolf protocol of~\cite{abeyesinghe_mother_2009}, will be presented elsewhere, but may be anticipated from~\cite{boileau_optimal_2009}.

The general approach in~\cite{renes_physical_2008,boileau_optimal_2009} is to first treat the simpler problems of enabling Bob to predict Alice's $Z^A$ and $X^A$ outcomes separately, and then figure out how to combine these later. In fact, each is just the task of classical data compression with quantum side information at the decoder; Alice's classical message to Bob is the compressed version of the classical random variable $Z^A$ ($X^A$), and system $B$ is the quantum side information that Bob, the decoder, uses to reconstruct $Z^A$ ($X^A$). 

Constructing optimal protocols for this task was done in~\cite{devetak_classical_2003}. To build an entanglement distillation protocol these two protocols simply need to be combined in some way, and fortunately this is possible with linear error-correcting codes. In the resulting entanglement distillation or state merging protocol Alice makes use of Calderbank-Shor-Steane (CSS) quantum error-correcting codes~\cite{calderbank_good_1996,steane_multiple-particle_1996} while Bob uses the decoder described here, built from the measurements for the $Z^A$ and $X^A$ reconstruction protocols. Note that the protocol proceeds in two steps, just like Theorem 1, so Bob's $X^A$ measurement can take advantage of the ancillary system $C$, and must do so to make an optimal protocol.

In this approach, the observables $Z^A$ and $X^A$ used in the eventual entanglement decoder are closely related (by the linear stabilizers of the CSS code) to the $Z^A$ and $X^A$ operators of the input state to the protocol. However, there is no need for this to be the case, all that is necessary is to find \emph{some} observables $X^A$ and $Z^A$ which Bob can predict. Hayden, Shor, and Winter give another means for doing this  in the course of proving the quantum noisy channel coding theorem in~\cite{hayden_random_2008}. Their approach is to to build up a subspace by picking vectors according to a Gaussian distribution. Since $X^A$ and $Z^A$ are related by the Fourier transform $F=\frac{1}{\sqrt{d}}\sum_{jk}\omega^{jk}\ket{j}\bra{k}$, this has the advantage that $p_{\rm guess}(Z^A|B)=p_{\rm guess}(X^A|B)$ automatically. It should be noted that their
proof of the noisy channel coding theorem could be simplified by using the entanglement decoder described here, as they take a more complicated route via decoupling. 

Finally, the optimality of the protocols mentioned above pertains to the rate at which entangled pairs can be produced in the limit of infinitely many copies of the input resource $\psi^{AB}$ or the necessary rate of classical communication, not the difficultly of implementing the decoding map. The entanglement recovery conditions established here give a means of investigating suboptimal but more efficient recovery maps by reducing the fully quantum problem into easier classical-quantum pieces.    

\emph{Summary and Relation to Previous Work.---} This paper provides two new characterizations of when entanglement recovery is possible from a bipartite state by local operations, and strengthens the intuition that quantum information is a fusion of classical information pertaining to two complementary observables.
Work on approximate quantum error correction falls into two broad camps, investigations into when entanglement recovery is possible and constructions of recovery operators. Furthermore, the former category contains both information-theoretic as well as algebraic approaches. This paper is of the information-theoretic variety, as was the original result~\cite{schumacher_approximate_2002}, showing that negligible loss of a quantity called coherent information implies entanglement recovery. 

Recently, Buscemi has extended this to negligible loss of entanglement of formation~\cite{buscemi_entanglement_2008}. Devetak and Winter implicitly used a condition similar to the two presented here in~\cite{devetak_private_2005,devetak_distillation_2005}, essentially a hybrid of them which delivers the same conclusion of Theorems 1 and 2 from the premises $p_{\rm guess}(Z^A|B)_\psi>1-\epsilon_z$ and $p_{\rm secure}(Z^A|R)_\psi\geq 1-\tfrac{1}{2}\epsilon_x^2$. That approach is closely related to quantum cryptography, as the two conditions state that $Z^A$ is an approximate secret key, a uniform random variable shared by Alice and Bob but uncorrelated with any other system.

On the algebraic side, Klesse~\cite{klesse_approximate_2007} adapted the Schumacher and Westmoreland result to find conditions for approximate quantum error correction and used this to give a proof of the quantum noiseless coding theorem~\cite{klesse_random_2008}. B\'eny and Oreshkov~\cite{bny_general_2009} link the optimal recoverability from a given quantum channel to the dual problem of recoverability from the complementary channel, and use this to give an approximate form of the original Knill-Laflamme conditions~\cite{knill_theory_1997} on perfect recovery~\footnote{By the Stinespring dilation, any channel $\mathcal{E}$ can be thought of as a unitary acting on a larger state space, followed by tracing out these extra degrees of freedom; the complementary channel $\widehat{\mathcal{E}}$ results when tracing out the original degrees of freedom.}. They also use this to construct recovery maps in certain cases. Independently, Ng and Mandayam~\cite{ng_simple_2009} derived similar different conditions on approximate error-correction from their use of the transpose map as a recovery operator. 

Because the decoupling argument only asserts a suitable recovery map exists, it is difficult to connect the usual information theoretic approaches with actual constructions for approximate quantum error correcting codes as is common in the algebraic approach. The first condition presented here does lend itself more readily to actual constructions, since one way to satisfy the two premises is to find measurements with high guessing probability. Moreover, the full quantum task is broken down into two easier classical-quantum tasks, and I hope that this work will help bridge the divide and spur new code constructions from information-theoretic reasoning.

I thank Mark M.\ Wilde and Andreas Winter for helpful discussions and the financial support of the Center for Advanced Security Research Darmstadt (CASED).
\bibliographystyle{apsrev}
\bibliography{caqec}

\end{document}

%% file: entdecfigure.tex
\newcommand{\ctrol}{\fill(0,0.1) circle(2pt);}
\newcommand{\target}{\draw(0,0.1) circle(4pt);}
\def\vgap{.75}
\def\hgap{1}
\def\hgapfudge{.25}
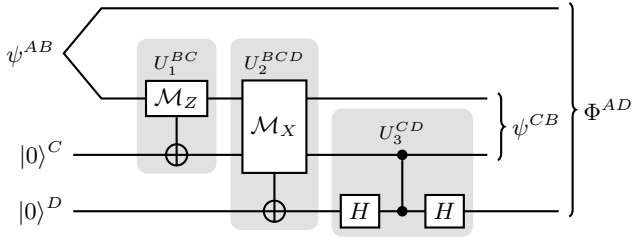
\begin{figure}[h!]
\begin{center}
\begin{tikzpicture}[thick]
\tikzstyle{empty} = [inner sep=1pt,outer sep=1pt]
\tikzstyle{gate} = [fill=white, draw]
\tikzstyle{ctrl} = [fill,shape=circle,minimum size=4pt,inner sep=0pt,outer sep=0pt]
\tikzstyle{targ} = [draw,shape=circle,minimum size=8pt,inner sep=0pt,outer sep=0pt]

\node[anchor=east] at (-0.5*\hgap,0.8*\vgap) (ab0) {$\psi^{AB}$};
\node[anchor=east] at (-0.375*\hgap,-\vgap) (cl) {$\ket{0}^C$};
\node[anchor=east] at (-0.375*\hgap,-2*\vgap) (dl) {$\ket{0}^D$};

\node[empty] at (6.15*\hgap,1.6*\vgap) (ar) {};
\node[empty] at (5.2*\hgap,0) (br) {};
\node[empty] at (5.2*\hgap,-\vgap) (cr) {};
\node[empty] at (6.15*\hgap,-2*\vgap) (dr) {};


\draw (br.west) -- (0,0) -- (-0.5*\hgap,0.8*\vgap) -- (0,1.6*\vgap) -- (ar.west);
\draw (cl.east) -- (cr.west);
\draw (dl.east) -- (dr.west);


\node[gate] at (\hgap,0) (mz) {$\mathcal{M}_Z$};
\node[targ] at (\hgap,-\vgap) (mzt) {};
\node[above=0 of mz,empty] (mzl) {\scriptsize $U_1^{BC}$};

\node[gate,minimum height=35pt] at (2*\hgap+\hgapfudge+.05,-.5*\vgap) (mx) {$\mathcal{M}_X$};
\node[targ] at (2*\hgap+\hgapfudge+.05,-2*\vgap) (mxt) {};
\node[above=0 of mx,empty] (mxl) {\scriptsize $U_2^{BCD}$};
\node[gate] at (3*\hgap+1.75*\hgapfudge,-2*\vgap) (h1) {$H$};

\node[ctrl] at (4*\hgap,-\vgap) (ctrlc) {};
\node[ctrl] at (4*\hgap,-2*\vgap) (ctrld) {};
\node[above=0 of ctrlc,empty] (ctrlcl) {\scriptsize $U_3^{CD}$};
\node[gate] at (5*\hgap-1.75*\hgapfudge,-2*\vgap) (h2) {$H$};



\draw (mz.south) -- (mzt.south);
\draw (mx.south) -- (mxt.south);
\draw (ctrlc) -- (ctrld);

\draw[decorate,decoration={brace},thick] (br.north east) to node[midway,right,inner sep=5pt] {$\psi^{CB}$} (cr.south east);
\draw[decorate,decoration={brace},thick] (ar.north east) to node[midway,right,inner sep=5pt] {$\Phi^{AD}$} (dr.south east);

\begin{pgfonlayer}{background}
\node [fill=black!12!white,rounded corners, dashed,fit=(h1) (ctrlc) (h2) (ctrlcl)] {};
\node [fill=black!12!white,rounded corners,fit=(mz) (mzt) (mzl)] {};
\node [fill=black!12!white,rounded corners,fit=(mx) (mxt) (mxl)] {};
\end{pgfonlayer}

\end{tikzpicture}
\caption{\label{fig:entdec} The quantum circuit enabling entanglement recovery from a bipartite state $\psi^{AB}$ by Bob, when he can approximately predict measurement of either conjugate observable $X$ or $Z$ by Alice. It proceeds in three steps, indicated in gray. First, Bob coherently performs the measurement $\mathcal{M}_Z$ allowing him to predict $Z$, storing the result in auxiliary system $C$ (unitary $U_{1}^{BC}$). Next, he coherently performs the measurement $\mathcal{M}_X$ allowing him to predict $X$, storing the result in auxiliary system $D$ (unitary $U_2^{BCD}$); knowledge of $Z$, now stored in $C$, may be necessary for this task. Finally, to recover a maximally entangled state in system $D$, he applies a controlled phase gate, with $D$ in the $x$ basis (unitary $U_3^{CD}$). This procedure also leaves Bob holding the original input state $\psi^{AB}$ in systems $C$ and $B$.}
\end{center}
\vspace{-14pt}
\end{figure}

%% file: caqec.bbl
\begin{thebibliography}{27}
\expandafter\ifx\csname natexlab\endcsname\relax\def\natexlab#1{#1}\fi
\expandafter\ifx\csname bibnamefont\endcsname\relax
  \def\bibnamefont#1{#1}\fi
\expandafter\ifx\csname bibfnamefont\endcsname\relax
  \def\bibfnamefont#1{#1}\fi
\expandafter\ifx\csname citenamefont\endcsname\relax
  \def\citenamefont#1{#1}\fi
\expandafter\ifx\csname url\endcsname\relax
  \def\url#1{\texttt{#1}}\fi
\expandafter\ifx\csname urlprefix\endcsname\relax\def\urlprefix{URL }\fi
\providecommand{\bibinfo}[2]{#2}
\providecommand{\eprint}[2][]{\url{#2}}

\bibitem[{\citenamefont{Schumacher and
  Westmoreland}(2002)}]{schumacher_approximate_2002}
\bibinfo{author}{\bibfnamefont{B.}~\bibnamefont{Schumacher}} \bibnamefont{and}
  \bibinfo{author}{\bibfnamefont{M.~D.} \bibnamefont{Westmoreland}},
  \bibinfo{journal}{Quant. Inf. Proc.} \textbf{\bibinfo{volume}{1}},
  \bibinfo{pages}{5} (\bibinfo{year}{2002}).

\bibitem[{\citenamefont{Klesse}(2007)}]{klesse_approximate_2007}
\bibinfo{author}{\bibfnamefont{R.}~\bibnamefont{Klesse}},
  \bibinfo{journal}{Phys. Rev. A} \textbf{\bibinfo{volume}{75}},
  \bibinfo{pages}{062315} (\bibinfo{year}{2007}).

\bibitem[{\citenamefont{Hayden et~al.}(2008{\natexlab{a}})\citenamefont{Hayden,
  Horodecki, Winter, and Yard}}]{hayden_decoupling_2008}
\bibinfo{author}{\bibfnamefont{P.}~\bibnamefont{Hayden}},
  \bibinfo{author}{\bibfnamefont{M.}~\bibnamefont{Horodecki}},
  \bibinfo{author}{\bibfnamefont{A.}~\bibnamefont{Winter}}, \bibnamefont{and}
  \bibinfo{author}{\bibfnamefont{J.}~\bibnamefont{Yard}},
  \bibinfo{journal}{Open Syst. Inf. Dyn.} \textbf{\bibinfo{volume}{15}},
  \bibinfo{pages}{7} (\bibinfo{year}{2008}{\natexlab{a}}).

\bibitem[{\citenamefont{Hayden et~al.}(2008{\natexlab{b}})\citenamefont{Hayden,
  Shor, and Winter}}]{hayden_random_2008}
\bibinfo{author}{\bibfnamefont{P.}~\bibnamefont{Hayden}},
  \bibinfo{author}{\bibfnamefont{P.~W.} \bibnamefont{Shor}}, \bibnamefont{and}
  \bibinfo{author}{\bibfnamefont{A.}~\bibnamefont{Winter}},
  \bibinfo{journal}{Open Syst. Inf. Dyn.} \textbf{\bibinfo{volume}{15}},
  \bibinfo{pages}{71} (\bibinfo{year}{2008}{\natexlab{b}}).

\bibitem[{\citenamefont{Horodecki et~al.}(2008)\citenamefont{Horodecki, Lloyd,
  and Winter}}]{horodecki_quantum_2008}
\bibinfo{author}{\bibfnamefont{M.}~\bibnamefont{Horodecki}},
  \bibinfo{author}{\bibfnamefont{S.}~\bibnamefont{Lloyd}}, \bibnamefont{and}
  \bibinfo{author}{\bibfnamefont{A.}~\bibnamefont{Winter}},
  \bibinfo{journal}{Open Syst. Inf. Dyn.} \textbf{\bibinfo{volume}{15}},
  \bibinfo{pages}{47} (\bibinfo{year}{2008}).

\bibitem[{\citenamefont{Klesse}(2008)}]{klesse_random_2008}
\bibinfo{author}{\bibfnamefont{R.}~\bibnamefont{Klesse}},
  \bibinfo{journal}{Open Syst. Inf. Dyn.} \textbf{\bibinfo{volume}{15}},
  \bibinfo{pages}{24} (\bibinfo{year}{2008}).

\bibitem[{\citenamefont{Bjelakovi\'c et~al.}(2009)\citenamefont{Bjelakovi\'c,
  Boche, and N\"otzel}}]{bjelakovi_entanglement_2009}
\bibinfo{author}{\bibfnamefont{I.}~\bibnamefont{Bjelakovi\'c}},
  \bibinfo{author}{\bibfnamefont{H.}~\bibnamefont{Boche}}, \bibnamefont{and}
  \bibinfo{author}{\bibfnamefont{J.}~\bibnamefont{N\"otzel}},
  \bibinfo{journal}{Comm. Math. Phys.} \textbf{\bibinfo{volume}{292}},
  \bibinfo{pages}{55} (\bibinfo{year}{2009}).

\bibitem[{\citenamefont{Buscemi and Datta}(2009)}]{buscemi_one-shot_2009}
\bibinfo{author}{\bibfnamefont{F.}~\bibnamefont{Buscemi}} \bibnamefont{and}
  \bibinfo{author}{\bibfnamefont{N.}~\bibnamefont{Datta}},
  \bibinfo{journal}{ar{X}iv:0902.0158}  (\bibinfo{year}{2009}).

\bibitem[{\citenamefont{Horodecki et~al.}(2005)\citenamefont{Horodecki,
  Oppenheim, and Winter}}]{horodecki_partial_2005}
\bibinfo{author}{\bibfnamefont{M.}~\bibnamefont{Horodecki}},
  \bibinfo{author}{\bibfnamefont{J.}~\bibnamefont{Oppenheim}},
  \bibnamefont{and} \bibinfo{author}{\bibfnamefont{A.}~\bibnamefont{Winter}},
  \bibinfo{journal}{Nature} \textbf{\bibinfo{volume}{436}},
  \bibinfo{pages}{673} (\bibinfo{year}{2005}).

\bibitem[{\citenamefont{Horodecki et~al.}(2007)\citenamefont{Horodecki,
  Oppenheim, and Winter}}]{horodecki_quantum_2007}
\bibinfo{author}{\bibfnamefont{M.}~\bibnamefont{Horodecki}},
  \bibinfo{author}{\bibfnamefont{J.}~\bibnamefont{Oppenheim}},
  \bibnamefont{and} \bibinfo{author}{\bibfnamefont{A.}~\bibnamefont{Winter}},
  \bibinfo{journal}{Comm. Math. Phys.} \textbf{\bibinfo{volume}{269}},
  \bibinfo{pages}{107} (\bibinfo{year}{2007}).

\bibitem[{\citenamefont{Abeyesinghe et~al.}(2009)\citenamefont{Abeyesinghe,
  Devetak, Hayden, and Winter}}]{abeyesinghe_mother_2009}
\bibinfo{author}{\bibfnamefont{A.}~\bibnamefont{Abeyesinghe}},
  \bibinfo{author}{\bibfnamefont{I.}~\bibnamefont{Devetak}},
  \bibinfo{author}{\bibfnamefont{P.}~\bibnamefont{Hayden}}, \bibnamefont{and}
  \bibinfo{author}{\bibfnamefont{A.}~\bibnamefont{Winter}},
  \bibinfo{journal}{Proc. R. Soc. A} \textbf{\bibinfo{volume}{465}},
  \bibinfo{pages}{2537} (\bibinfo{year}{2009}).

\bibitem[{\citenamefont{Bennett et~al.}(1996)\citenamefont{Bennett,
  {DiVincenzo}, Smolin, and Wootters}}]{bennett_mixed-state_1996}
\bibinfo{author}{\bibfnamefont{C.~H.} \bibnamefont{Bennett}},
  \bibinfo{author}{\bibfnamefont{D.~P.} \bibnamefont{{DiVincenzo}}},
  \bibinfo{author}{\bibfnamefont{J.~A.} \bibnamefont{Smolin}},
  \bibnamefont{and} \bibinfo{author}{\bibfnamefont{W.~K.}
  \bibnamefont{Wootters}}, \bibinfo{journal}{Phys. Rev. A}
  \textbf{\bibinfo{volume}{54}}, \bibinfo{pages}{3824} (\bibinfo{year}{1996}).

\bibitem[{\citenamefont{Deutsch et~al.}(1996)\citenamefont{Deutsch, Ekert,
  Jozsa, Macchiavello, Popescu, and Sanpera}}]{deutsch_quantum_1996}
\bibinfo{author}{\bibfnamefont{D.}~\bibnamefont{Deutsch}},
  \bibinfo{author}{\bibfnamefont{A.}~\bibnamefont{Ekert}},
  \bibinfo{author}{\bibfnamefont{R.}~\bibnamefont{Jozsa}},
  \bibinfo{author}{\bibfnamefont{C.}~\bibnamefont{Macchiavello}},
  \bibinfo{author}{\bibfnamefont{S.}~\bibnamefont{Popescu}}, \bibnamefont{and}
  \bibinfo{author}{\bibfnamefont{A.}~\bibnamefont{Sanpera}},
  \bibinfo{journal}{Phys. Rev. Lett.} \textbf{\bibinfo{volume}{77}},
  \bibinfo{pages}{2818} (\bibinfo{year}{1996}).

\bibitem[{\citenamefont{Lo and Chau}(1999)}]{lo_unconditional_1999}
\bibinfo{author}{\bibfnamefont{H.}~\bibnamefont{Lo}} \bibnamefont{and}
  \bibinfo{author}{\bibfnamefont{H.~F.} \bibnamefont{Chau}},
  \bibinfo{journal}{Science} \textbf{\bibinfo{volume}{283}},
  \bibinfo{pages}{2050} (\bibinfo{year}{1999}).

\bibitem[{\citenamefont{Shor and Preskill}(2000)}]{shor_simple_2000}
\bibinfo{author}{\bibfnamefont{P.~W.} \bibnamefont{Shor}} \bibnamefont{and}
  \bibinfo{author}{\bibfnamefont{J.}~\bibnamefont{Preskill}},
  \bibinfo{journal}{Phys. Rev. Lett.} \textbf{\bibinfo{volume}{85}},
  \bibinfo{pages}{441} (\bibinfo{year}{2000}).

\bibitem[{\citenamefont{Devetak and Winter}(2005)}]{devetak_distillation_2005}
\bibinfo{author}{\bibfnamefont{I.}~\bibnamefont{Devetak}} \bibnamefont{and}
  \bibinfo{author}{\bibfnamefont{A.}~\bibnamefont{Winter}},
  \bibinfo{journal}{Proc. R. Soc. A} \textbf{\bibinfo{volume}{461}},
  \bibinfo{pages}{207} (\bibinfo{year}{2005}).

\bibitem[{\citenamefont{Renes}(2010)}]{renes_duality_2010}
\bibinfo{author}{\bibfnamefont{J.~M.} \bibnamefont{Renes}},
  \bibinfo{journal}{ar{X}iv:1003.0703}  (\bibinfo{year}{2010}).

\bibitem[{\citenamefont{Renes and Boileau}(2008)}]{renes_physical_2008}
\bibinfo{author}{\bibfnamefont{J.~M.} \bibnamefont{Renes}} \bibnamefont{and}
  \bibinfo{author}{\bibfnamefont{J.}~\bibnamefont{Boileau}},
  \bibinfo{journal}{Phys. Rev. A} \textbf{\bibinfo{volume}{78}},
  \bibinfo{pages}{032335} (\bibinfo{year}{2008}).

\bibitem[{\citenamefont{Boileau and Renes}(2009)}]{boileau_optimal_2009}
\bibinfo{author}{\bibfnamefont{J.}~\bibnamefont{Boileau}} \bibnamefont{and}
  \bibinfo{author}{\bibfnamefont{J.~M.} \bibnamefont{Renes}}, in
  \emph{\bibinfo{booktitle}{Fourth Workshop, {TQC} 2009}}
  (\bibinfo{publisher}{Springer, Berlin}, \bibinfo{address}{Waterloo, Canada},
  \bibinfo{year}{2009}), vol. \bibinfo{volume}{5906} of
  \emph{\bibinfo{series}{Lecture Notes in Computer Science}},
  p.~\bibinfo{pages}{76}.

\bibitem[{\citenamefont{Devetak and Winter}(2003)}]{devetak_classical_2003}
\bibinfo{author}{\bibfnamefont{I.}~\bibnamefont{Devetak}} \bibnamefont{and}
  \bibinfo{author}{\bibfnamefont{A.}~\bibnamefont{Winter}},
  \bibinfo{journal}{Phys. Rev. A} \textbf{\bibinfo{volume}{68}},
  \bibinfo{pages}{042301} (\bibinfo{year}{2003}).

\bibitem[{\citenamefont{Calderbank and Shor}(1996)}]{calderbank_good_1996}
\bibinfo{author}{\bibfnamefont{A.~R.} \bibnamefont{Calderbank}}
  \bibnamefont{and} \bibinfo{author}{\bibfnamefont{P.~W.} \bibnamefont{Shor}},
  \bibinfo{journal}{Phys. Rev. A} \textbf{\bibinfo{volume}{54}},
  \bibinfo{pages}{1098} (\bibinfo{year}{1996}).

\bibitem[{\citenamefont{Steane}(1996)}]{steane_multiple-particle_1996}
\bibinfo{author}{\bibfnamefont{A.}~\bibnamefont{Steane}},
  \bibinfo{journal}{Proc. R. Soc. A} \textbf{\bibinfo{volume}{452}},
  \bibinfo{pages}{2551} (\bibinfo{year}{1996}).

\bibitem[{\citenamefont{Buscemi}(2008)}]{buscemi_entanglement_2008}
\bibinfo{author}{\bibfnamefont{F.}~\bibnamefont{Buscemi}},
  \bibinfo{journal}{Phys. Rev. A} \textbf{\bibinfo{volume}{77}},
  \bibinfo{pages}{012309} (\bibinfo{year}{2008}).

\bibitem[{\citenamefont{Devetak}(2005)}]{devetak_private_2005}
\bibinfo{author}{\bibfnamefont{I.}~\bibnamefont{Devetak}},
  \bibinfo{journal}{IEEE Trans. Inf. Theory} \textbf{\bibinfo{volume}{51}},
  \bibinfo{pages}{44} (\bibinfo{year}{2005}).

\bibitem[{\citenamefont{B\'eny and Oreshkov}(2009)}]{bny_general_2009}
\bibinfo{author}{\bibfnamefont{C.}~\bibnamefont{B\'eny}} \bibnamefont{and}
  \bibinfo{author}{\bibfnamefont{O.}~\bibnamefont{Oreshkov}},
  \bibinfo{journal}{ar{X}iv:0907.5391}  (\bibinfo{year}{2009}).

\bibitem[{\citenamefont{Knill and Laflamme}(1997)}]{knill_theory_1997}
\bibinfo{author}{\bibfnamefont{E.}~\bibnamefont{Knill}} \bibnamefont{and}
  \bibinfo{author}{\bibfnamefont{R.}~\bibnamefont{Laflamme}},
  \bibinfo{journal}{Phys. Rev. A} \textbf{\bibinfo{volume}{55}},
  \bibinfo{pages}{900} (\bibinfo{year}{1997}).

\bibitem[{\citenamefont{Ng and Mandayam}(2009)}]{ng_simple_2009}
\bibinfo{author}{\bibfnamefont{H.~K.} \bibnamefont{Ng}} \bibnamefont{and}
  \bibinfo{author}{\bibfnamefont{P.}~\bibnamefont{Mandayam}},
  \bibinfo{journal}{ar{X}iv:0909.0931}  (\bibinfo{year}{2009}).

\end{thebibliography}
